%% file: main.tex
\documentclass[runningheads]{llncs}
\usepackage[T1]{fontenc}

\usepackage{graphicx}				

\usepackage{url} 
\usepackage{amsmath,amssymb}
\usepackage{bbm}
\usepackage{bm}
\usepackage{hyperref}
\usepackage{scrextend}

\usepackage{float}
\usepackage{textcomp}
\usepackage{pgfplots}
\pgfplotsset{compat=1.16} 
\usepackage{graphicx}
\usepackage{wrapfig}
\usepackage{subcaption}
\usepackage[]{youngtab}
\usepackage{tikz}
\usepackage{tikz-cd}
\usetikzlibrary{automata, positioning}
\usetikzlibrary{decorations.pathreplacing}
\usepackage{comment}
\usepackage{physics}  
\usepackage{stmaryrd}
\usepackage{theoremref}
\usepackage[toc,page]{appendix}

\usepackage{mathrsfs}

\usepackage{xspace}
\usepackage{enumitem}

\newcommand{\arrival}{\mathcal{A}}

\usepackage{thm-restate}

\newcommand{\E}{\mathbb{E}}

\newcommand{\qedwhite}{\hfill \ensuremath{\Box}}

\newcommand{\policybest}{$\pi^*$\xspace}

\title{How to Rent GPUs on a Budget\vspace{-.2in}}
\author{\vspace{-0.05in} Zhouzi Li\inst{1} \and
Benjamin Berg \inst{3} \and
Arpan Mukhopadhyay\inst{2} \and
Mor Harchol-Balter\inst{1}
}
\authorrunning{Li et al.}
\institute{Carnegie Mellon University \and
Warwick University \and
University of North Carolina at Chapel Hill
}
\begin{document}
\maketitle
\vspace{-.2in}

\input{EPEW24/epew24}

\section{Acknowledgement}
This work was funded by NSF-CMMI-2307008 and NSF-III-2322973


%
%
%

\bibliographystyle{splncs04}
\bibliography{bib,bibshort}

\end{document}

%% file: EPEW24/epew24.tex
\input{EPEW24/intro}

\input{EPEW24/Analysis}
\input{EPEW24/conclusion}

%% file: EPEW24/intro.tex
\vspace{-.05in}
\section{Introduction}
\vspace{-.05in}
\label{sec:intro}
The explosion in Machine Learning (ML) over the past ten years has led to a dramatic increase in demand for GPUs to train ML models \cite{jouppi2017datacenter}.
Because it is prohibitively expensive for most users to build and maintain a large GPU cluster, large cloud providers (Microsoft Azure, Amazon AWS, Google Cloud) have seen explosive growth in demand for renting cloud-based GPUs.
For example, Amazon's profit from renting GPUs and related infrastructure is expected to grow to 20 billion dollars by 2026, roughly 20\% of the profit currently generated by AWS~\cite{gpureport}.
This trend makes the question of how to efficiently train ML models using cloud-based infrastructure one of the most important problems facing the computer systems community\cite{subramanya2023sia,qiao2021pollux,moritz2018ray,misra2021rubberband}.

In this cloud computing paradigm, the user is responsible for devising a \emph{rental policy} that decides how many GPUs to rent at every moment in time\cite{salvaris2018microsoft}.
The cloud provider's goal is to satisfy all user resource demands with high probability.
ML training jobs can be parallelized across different numbers of GPUs, so the user generally has many options for how many GPUs to use for each of their jobs\cite{qiao2021pollux,subramanya2023sia}.
Allocating more GPUs to a single training job will cause the job to complete more quickly.
However, the user pays for each GPU-hour they use, and training jobs receive a diminishing marginal benefit from running on additional GPUs\cite{lin2018model}.
Hence, allocating too many GPUs to a single training job can dramatically increase the overall cost that the user pays to the cloud provider\cite{zheng2019cynthia}.
This raises the question of how a user can complete training jobs quickly without spending too much.
Despite the ubiquity of model training in the cloud, relatively little is understood about what rental policy a user should employ.

\vspace{3mm}
\noindent\textbf{The Problem.}
We consider the case where the user wants to minimize the average response time across a stream of arriving ML training jobs --- the average time from when a job is submitted until it is completed.
Each training job has its own level of parallelizability and \emph{job size} (amount of inherent work) that reflect the particular model and the training data being used.
To complete their jobs, each user has a limited \emph{operating budget} that they are willing to pay per hour to train the models used by their application.
For example, a user might need to train 5 models each hour, and they might be willing to spend \$50 per hour on average to complete these training jobs.
The user's goal is to minimize the mean response time across their jobs while respecting their operating budget.

To solve this problem, the user must decide both {\em (i)} how many GPUs to rent at every moment in time, and {\em (ii)} how to divide these GPUs across the training jobs that are present at that moment.  For example, how many additional GPUs should be rented during times of higher load?   How should GPUs be split between jobs with smaller size (small inherent work)  versus  those of larger size?
How should GPUs be split between jobs with higher parallelizability versus less parallelizable jobs?  At present, the user has little guidance on how to make these decisions.  The goal of this paper is to provide an optimal rental policy for the user as a function of their workload and their operating budget.  

For every job in the system, a rental policy must balance a tradeoff between the response time of the job and the cost of training the job.
This is particularly difficult because each job is allowed to have a different \emph{speedup function} that describes how fast the job runs as a function of the number of GPUs it is allocated.
Furthermore, we allow for an arbitrary arrival process of training jobs into the system, and allow job sizes to follow an arbitrary distribution.

\vspace{3mm}
\noindent\textbf{Our Model.}
We consider a user who submits a stream of ML training jobs to the system over time.
Our user has $M$ types of jobs. 
Each job type is associated with a job size distribution $X_i$ and a speedup function $s_i(k)$.
If a type-$i$ job of size $x$ runs on $k\geq 1$ GPUs, it will complete in time $\frac{x}{s_i(k)}$.
We assume $k$ can be fractional, but prior work has shown how to convert a policy with fractional assignments into a policy that only allocates whole GPUs \cite{ghanbarian2023performance}.

We make some mild assumptions about the form of the speedup functions. 
For any job type $i$, $s_i(k)$ should fulfill the following properties:
\begin{itemize}
    \item $s_i(k)$ is defined on $k \in [1,+\infty)$ and is continuous;
    \item {\em Monotonicity}:  $s_i(k_1)\leq s_i(k_2)$ for any $1\leq k_1<k_2$;
    \item {\em Concavity}:  $s_i(k_1)/k_1 \geq s_i(k_2)/k_2$ for any $1\leq k_1<k_2$.
\end{itemize}

We assume type-$i$ jobs arrive with mean rate $\lambda_i$ and have mean job size $\E[X_i]$.
If $N_i(t)$ denotes the number of type-$i$ arrivals by time $t$, we assume 
\begin{equation}
    \lambda_i = \lim_{t \to \infty} \frac{N_i(t)}{t} \quad \mbox{w.p. 1} \quad\mbox{and}\quad \E[X_i]= \lim_{t \to \infty} \frac{\sum_{j=1}^{N_i(t)}X_{ij}}{N_i(t)} \quad \mbox {w.p. 1}.
    \label{eq:well-behaved}
\end{equation}

We do {\em not} assume that job sizes are independent, nor do we assume that the inter-arrival times are independent. 
For convenience, we define the \emph{system load} contributed by type-$i$ jobs to be $\rho_i=\lambda_i\E[X_i]$.

The user's goal is to minimize the mean response time across all jobs, subject to the time-average budget constraint. 
Let $T_j$ denote the response time of the $j$th job, the time from when the job arrives until it is complete.  Then
\begin{eqnarray*}
\mbox{Average response time } = \E[T] & = & \lim_{t \to \infty} \frac{\sum_{j=1}^{N(t)} T_j}{N(t)}. \label{eqn:meanresponsetime}
\end{eqnarray*}

The user must maintain a time-average operating budget $\bar{B} < b$.  We assume there is a fixed cost to rent a single GPU for an hour, so we can state the operating budget as a time-average number of GPUs that are rented.   Mathematically, if $K(s)$ denotes the number of GPUs rented at time $s$, then
\begin{eqnarray*}
\mbox{Time-average budget} = \bar{B} & = & \lim_{t \to \infty} \frac{\int_0^t K(s) ds}{t}. \label{eqn:budget}
\end{eqnarray*}
We restrict our analysis to rental policies such that $\bar{B}$ exists.
Finally, we assume that the budget limit, $b$, is sufficiently high to allow the system to be stable:
\begin{equation}
\label{assumption:stability}
    \sum_{i=1}^M \rho_i < b.
\end{equation}
If this assumption is violated, the work in the system will explode: the left hand side is the average rate of work coming into the system, while the right hand side is the maximum average rate of completing work given a time average budget limit $b$ (because of the sub-linearity of the speedup functions).



\vspace{3mm}
\noindent\textbf{Our Results.}
In this extremely general model, we show that under very mild technical assumptions, the optimal scheduling policy for jobs is surprisingly simple.
We begin in Section \ref{sec:offline} by considering an offline variant of the problem, where job sizes and arrival times are known \emph{a priori}.
We prove that, in this offline case, the optimal policy should never queue jobs, and that the optimal policy is a {\em fixed-width} policy that runs every type-$i$ job (regardless of its size) on some fixed number of GPUs, $k_i$ ($k_i$ can be computed via a convex optimization problem).
In Section \ref{sec:online} we argue that, because the offline optimal policy can be implemented online, it is also optimal for the online problem.

Figure~\ref{fig:pic} shows our optimal policy applied to two types of jobs.  In Figure~\ref{fig:pic}(a), we see the two speedup functions.  In Figure~\ref{fig:pic}(b), we see the optimal allocation of GPUs to each type.  Finally, Figure~\ref{fig:pic}(c) shows the Pareto frontier between mean response time and budget for the optimal policy.

\begin{figure}[h]
    \centering
        \begin{subfigure}{.3\textwidth}
        \centering
        \includegraphics[width=1\linewidth]{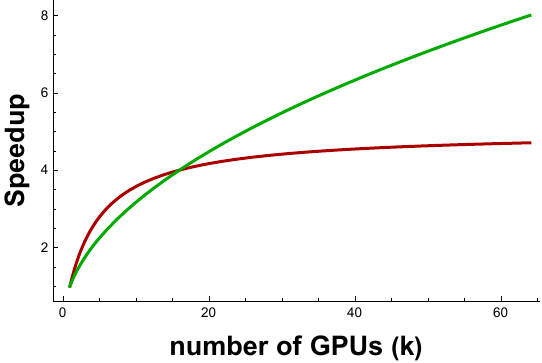}
        \caption{Speed up functions.}
    \end{subfigure}
    \quad
    \begin{subfigure}{.3\textwidth}
        \centering
        \includegraphics[width=1\linewidth]{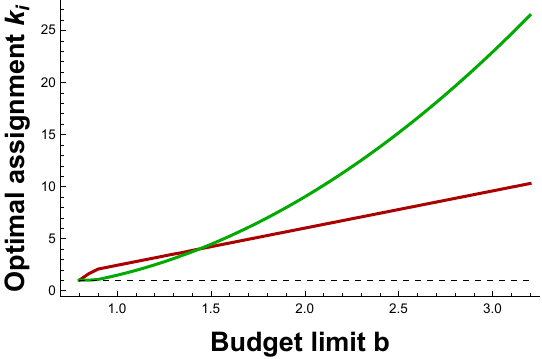}
        \caption{The optimal $k_i$.}
    \end{subfigure}
    \quad
    \begin{subfigure}{.3\textwidth}
        \centering
        \includegraphics[width=1\linewidth]{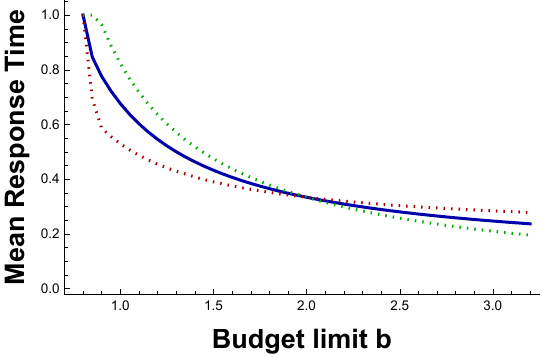}
        \caption{Pareto Frontier.}
    \end{subfigure}
    \caption{Example with two job types. Type 1 jobs (red) have a speedup function, $s_1(k)$, following Amdahl's Law, with 0.8 fraction parallelizable. Type 2 jobs (green) have a speedup function $s_2(k) = \sqrt{k}$. Both types have arrival rate $0.4$ and mean job size $1$.  The blue line in (c) shows the overall mean response time.}
    \label{fig:pic}
\end{figure}
\vspace{-.4in}

%% file: EPEW24/Analysis.tex
\section{Results and Analysis}
\label{sec:analysis}
In this section, we derive the optimal policy for the stochastic model that we introduced in Section~\ref{sec:intro}. We start by defining a {\em fixed width policy}.
We then define a particular fixed width policy, \policybest, and prove its optimality.

\begin{definition}[Fixed width policy]
\label{def:fixedwidth}
    A \textit{fixed width} policy chooses constants $k_1, k_2, \ldots k_M$.  All type-$i$ jobs are assigned $k_i$ GPUs immediately on arrival.  Jobs do not queue and do not change the number of GPUs they utilize as they run.
\end{definition}

\begin{definition}[Optimal policy, \policybest]
    \label{def:policybest}
Under the policy \policybest,
whenever a type-$i$ job comes into the system, it is immediately assigned $k_i$ cores until it is completed, where $k_i$ is the solution of the following optimization problem: 
\begin{equation}
        \begin{aligned}
        & \underset{k_i,i\in[M]} {\text{minimize}}
        & & 
        \sum_{i=1}^M\frac{  \rho_i }{s_i(k_i)}
        \\
        & \text{subject to}
        & & 
        \sum_{i=1}^M\frac{\rho_i k_i}{s_i(k_i)}\leq b
        ,\\
        \end{aligned}
    \label{eq:online opt}
\end{equation}
\end{definition}
 Note that \eqref{eq:online opt} is a feasible problem because setting all $k_i=1$ is feasible by the stability assumption \eqref{assumption:stability}. We can rewrite \eqref{eq:online opt} as a convex optimization problem, making it easy to solve numerically (we omit this for lack of space).

 Our main theorem is stated as below.

\begin{theorem}[\policybest is optimal]
   In our stochastic model, policy \policybest minimizes the mean response time among all online policies, where $\bar{B} \leq b$.
   \label{thm:main}
\end{theorem}

\noindent\textbf{Outline of our proof:}
The rest of this section is devoted to proving Theorem~\ref{thm:main}.
Our proof has two parts:  In Section~\ref{sec:offline}, we consider an offline variant of the problem and show that \policybest minimizes the mean response time for ``well-behaved'' offline arrival sequences.  
In Section~\ref{sec:online}, we return to the online problem, showing that \policybest is the best policy for the online setting as well.

\subsection{The offline problem}\label{sec:offline}
We now define an offline variant of our problem.
A sample path $\arrival$ is an infinite sequence of arrival times and job sizes, where $x_{ij}$ denotes the size of the $j$th arrival of type $i$.  
In our offline setting, we assume that job sizes and arrivals times are known to the system.
Let $n_i(t)$ denote the number of type-$i$ arrivals by time $t$, and $n(t):=\sum_{i=1}^M n_i(t)$ denote the total number of arrivals by time $t$. 

We say that the sample path $\arrival$ is {\em well-behaved} if its arrival process and job sizes both converge to the user's averages, as previously defined in \eqref{eq:well-behaved}:
\[\lambda_i = \lim_{t \to \infty} \frac{n_i(t)}{t}  \quad\mbox{and}\quad \E[X_i]= \lim_{t \to \infty} \frac{\sum_{j=1}^{n_i(t)}x_{ij}}{n_i(t)} \qquad \forall i.\]

\begin{lemma}[Fixed Width Offline]
\label{lemma:fixwidth}
For any well-behaved sample path $\arrival$, there exists a fixed width policy that minimizes the mean response time. 
\end{lemma}

\begin{proof}
First, we show that an optimal policy does not queue jobs.
Assume there is a job queueing in the optimal policy. Removing its waiting time makes the mean response time lower but leaves the total GPU-hours the same, a contradiction.

Next, we show that an optimal policy assigns the same number of GPUs to each type-$i$ job.
    Suppose the optimal policy $P$ is not fixed width for type-$i$ jobs. 
    Then either $P$ assigns $k_1\neq k_2$ GPUs to different type-$i$ jobs, or $P$ assigns $k_1\neq k_2$ GPUs to a single type-$i$ job at different times. Let $x_1,x_2$ be the work completed using $k_1$ and $k_2$ cores, respectively.
    

    We construct a policy $P'$ which uses the same GPU-hours, but has lower mean response time than $P$.
    $P'$ is identical to $P$ except for the first time that $P$ violates the Fixed Width definition.
    Instead of using two different assignments $k_1$ and $k_2$, $P'$ will choose a constant number of GPUs, $k$, to use in both instances.
    Let $t_1=\frac{x_1}{s_i(k_1)}$ and $t_2=\frac{x_2}{s_i(k_2)}$ be the durations of each of the GPU assignments under $P$.
    We choose $k$ to be the time average of the two assignments by setting
    \begin{equation}
        k= k_1\cdot \frac{t_1}{t_1+t_2} + k_2 \cdot \frac{t_2}{t_1+t_2}.
        \label{eq:fix-width proof eq1}
    \end{equation}

    The concavity of $s_i(k)$ implies that the total time to complete $x_1$ and $x_2$ is lower under $P'$.  
    To see this, consider the average rate of work completion when processing $x_1$ and $x_2$ under both policies.
    The average work rate under $P$ is 
    $$\frac{t_1}{t_1+t_2}\cdot s_i(k_1) + \frac{t_2}{t_1+t_2}\cdot s_i(k_2).$$
    The average work rate under $P'$ is $s_i(k)$.
    By concavity, we have
    $$\frac{t_1}{t_1+t_2}\cdot s_i(k_1) + \frac{t_2}{t_1+t_2}\cdot s_i(k_2)\leq s_i\left(\frac{t_1}{t_1+t_2}\cdot k_1 + \frac{t_2}{t_1+t_2}\cdot k_2\right)=s_i(k).$$
    The total time to process $x_1$ and $x_2$ can be computed as the total work ($x_1 +x_2$) divided by the average work rate.
    We thus have that $P'$ completes the $x_1 + x_2$ work sooner than $P$.
    As a result, $P'$ has a lower mean response time than $P$.

    It is easy to see that $P'$ does not use more total GPU-hours than $P$.
    Specifically, note that $P$ uses $k_1 t_1 + k_2 t_2$ GPU-hours to process $x_1$ and $x_2$.
    Let $t'_1$ and $t'_2$ be the times required to process $x_1$ and $x_2$ respectively under $P'$.  Then the GPU-hours used to process $x_1$ and $x_2$ under $P'$ is
    $$(t'_1 +t'_2)\left(\frac{k_1t_1}{t_1+t_2} + \frac{k_2t_2}{t_1+t_2}\right).$$
    We have already shown that $t'_1 + t'_2 \leq t_1+t_2$, giving 
    $$(t'_1 +t'_2)\left(\frac{k_1t_1}{t_1+t_2} + \frac{k_2t_2}{t_1+t_2}\right)\leq(t_1 +t_2)\left(\frac{k_1t_1}{t_1+t_2} + \frac{k_2t_2}{t_1+t_2}\right)= k_1t_1+k_2t_2.$$

    In summary, if we consider an optimal, non-fixed width policy, we can iteratively remove every violation of the fixed-width conditions without increasing the mean response time or violating the budget constraint.  
    Hence, there exists an optimal fixed width policy. $\qedwhite$
\end{proof}

Lemma~\ref{lemma:fixwidth} says that the optimal policy in the offline setting is a fixed width policy.    
To compute the optimal fixed width policy, it will be helpful to first develop an alternate formulation of the operating budget of a fixed width policy.

\begin{lemma}[Operating budget of a fixed width policy]
\label{lemma:budget of fixed width}
    Given a well-behaved sample path $\arrival$, the operating budget of a fixed width policy with parameters $k_1,k_2,\dots, k_M$ is
    \[\bar{B} :=\lim_{t\to\infty} \frac{\int_0^tK(s)ds}{t} \ \stackrel{(a)}{=} \ \lim_{t\to\infty} \frac{\sum_{i=1}^{n(t)}B^{(i)}}{t} \ \stackrel{(b)}{=}  \ \sum_{i=1}^M \frac{\rho_ik_i}{s_i(k_i)}.\]
    Here $B^{(i)}$ is defined to be the GPU-hours used to complete the $i^{th}$ arriving job.
\end{lemma}

\begin{proof}
    Here we provide a proof sketch --- the full proof is omitted for brevity.

    Part (a) of our claim says that tracking the total GPU usage at every time $t$ is equivalent to tracking the GPU-hours used to process each job, $B^{(i)}$. 

    To prove part (b), we show that we can take the limit of this equivalent formulation to prove our claim.
    Note that the fixed width policy assigns $k_i$ cores to any type-$i$ job. 
    Hence, the GPU-hours spent on the $j^{th}$ type-$i$ job is $\frac{x_{ij}k_i}{s_i(k_i)}$.

    Thus we have that 
    \[\lim_{t\to\infty} \frac{\sum_{i=1}^{n(t)}B^{(i)}}{t}=\lim_{t\to\infty} \frac{\sum_{i=1}^{M} \sum_{j=1}^{n_i(j)} \frac{x_{ij}k_i}{s_i(k_i)}}{t} =\sum_{i=1}^M \frac{k_i}{s_i(k_i)}\left(\lim_{t\to\infty}
        \frac{\sum_{j=1}^{n_i(j)}x_{ij}}{t}
        \right),\]
    where
    \[\lim_{t\to\infty}
        \frac{\sum_{j=1}^{n_i(j)}x_{ij}}{t} = \lim_{t\to\infty}
        \frac{\sum_{j=1}^{n_i(j)}x_{ij}}{n_i(t)} \frac{n_i(t)}{t} = \lambda_i\E[X_i] = \rho_i.\]
    
    \vspace{-.15in}
    $\qedwhite$
\end{proof}

We now show that \policybest is offline optimal for any well-behaved sample path $\arrival$. 

\begin{lemma}
\label{lemma:policybest}
    For any well-behaved sample path $\arrival$, \policybest is the optimal offline policy.
\end{lemma}
\begin{proof}
    Lemma~\ref{lemma:fixwidth} shows that the optimal offline policy is a fixed width policy. Thus, it suffices to show that \policybest is the optimal offline fixed width policy.

    For any job in $\arrival$ of type $i$ and size $x$, the response time under a fixed width policy parameterized with $k_i$ is $\frac{x}{s_i(k_i)}$. 
    Thus, we have that 
    \begin{align*}
        \E[T]:=\lim_{t\to\infty} \frac{\sum_{i=1}^{n(t)}T_i}{n(t)}=\lim_{t\to\infty}\sum_{i=1}^M \frac{\sum_{j=1}^{n_i(t)}x_{ij}}{s_i(k_i)n(t)}&=\sum_{i=1}^M\frac{1}{s_i(k_i)} \left(
    \lim_{t\to\infty} \frac{\sum_{j=1}^{n_i(t)}x_{ij}}{n(t)}
    \right)\\
    &=\frac{1}{\lambda}\sum_{i=1}^M \frac{\rho_i}{s_i(k_i)}.
    \end{align*}
    Moreover, by Lemma~\ref{lemma:budget of fixed width}, the operating budget is $\sum_{i=1}^M \frac{\rho_ik_i}{s_i(k_i)}$. Thus solving the optimal set of $k_1,k_2,...,k_M$ is solving the following optimization problem:
    \begin{equation*}
        \begin{aligned}
        & \underset{k_i,i\in[M]} {\text{minimize}}
        & & \frac{1}{\lambda}\sum_{i=1}^M\frac{  \rho_i }{s_i(k_i)}\\
        & \text{subject to}
        & & \sum_{i=1}^M\frac{\rho_i k_i}{s_i(k_i)}\leq b.\\
        \end{aligned}
    \end{equation*}

This is the same as the optimization problem (\ref{eq:online opt}) except for a constant $1/\lambda$. This shows that \policybest is the optimal fixed width policy. $\qedwhite$
\end{proof}

\subsection{Returning to the online problem}\label{sec:online}

We now prove our main theorem, Theorem~\ref{thm:main}. 
Lemma~\ref{def:policybest} shows that, for any well-behaved sample path, policy \policybest is offline optimal. Because \policybest only uses the values of the $\rho_i$'s (the load of each type), and not any specific job arrival times or sizes, \policybest is an online policy. Hence, \policybest is the optimal online policy for any well-behaved sample path.  
We assumed that the set of sample paths that are not well-behaved has measure 0, so \policybest minimizes mean response time.


%% file: EPEW24/conclusion.tex
\section{Conclusion}
We show that, when running ML training jobs in the cloud with a fixed operating budget, we can compute an optimal policy for renting GPUs.
While our policy does depend on the speedup functions of the different jobs, surprisingly, it does not prioritize short jobs over long ones, as suggested in prior work \cite{berg2020hesrpt}.